%% file: YukawaPotential.tex
\documentclass[aip,reprint,twocolumn,superscriptaddress,showkeys]{revtex4-1}

\usepackage{style}
\begin{document}

\title{Solution to the modified Helmholtz equation for arbitrary periodic charge densities}
\author{Miriam Hinzen}
\email[Corresponding author: ]{m.hinzen@fz-juelich.de}
\affiliation{Institute for Advanced Simulation, Forschungszentrum J\"ulich, 52425 J\"ulich, Germany}
\affiliation{J\"ulich Supercomputing Centre, Forschungszentrum J\"ulich, 52425 J\"ulich, Germany}
\affiliation{JARA-CSD, 52425 J\"ulich, Germany}
\affiliation{Peter-Gr\"unberg Institute, Forschungszentrum J\"ulich, 52425 J\"ulich, Germany}
\author{Edoardo Di Napoli}
\affiliation{Institute for Advanced Simulation, Forschungszentrum J\"ulich, 52425 J\"ulich, Germany}
\affiliation{J\"ulich Supercomputing Centre, Forschungszentrum J\"ulich, 52425 J\"ulich, Germany}
\affiliation{JARA-CSD, 52425 J\"ulich, Germany}
\author{Daniel Wortmann}
\author{Stefan Bl\"ugel}
\affiliation{Institute for Advanced Simulation, Forschungszentrum J\"ulich, 52425 J\"ulich, Germany}
\affiliation{JARA-CSD, 52425 J\"ulich, Germany}
\affiliation{Peter-Gr\"unberg Institute, Forschungszentrum J\"ulich, 52425 J\"ulich, Germany}
\date{\today}


\input{Abstract}

\maketitle

\input{Introduction}
\input{MainMath}
\input{Conclusion}
\input{Acknowledgments}
\input{Appendix}
\input{DataAvailability}
\bibliography{YukawaPotential}

\end{document}

%% file: Abstract.tex
\begin{abstract}
We present a general method for solving the modified Helmholtz equation without shape approximation for an arbitrary
periodic charge distribution, whose solution is known as the Yukawa potential or the screened Coulomb potential. The method is an extension of Weinert's pseudo-charge method [M.\ Weinert, J.\ Math.\ Phys.\ \textbf{22},\ 2433 (1981)] for solving the Poisson equation for the same class of charge density distributions. The inherent differences between the Poisson and the modified Helmholtz equation are in their respective radial solutions. These are polynomial functions, for the Poisson equation, and modified spherical Bessel functions, for the modified Helmholtz equation. This leads to a definition of a modified pseudo-charge density and modified multipole moments. We have shown that Weinert's convergence analysis of an absolutely and uniformly convergent Fourier series of the pseudo-charge density is transferred to the modified pseudo-charge density. We conclude by illustrating the algorithmic changes necessary to turn an available implementation of the Poisson solver into a solver for the modified Helmholtz equation.
\end{abstract}

%% file: Introduction.tex
\section{Introduction}

A variety of problems in condensed matter physics require an efficient solution of the partial differential equation
\begin{equation}\label{mod_Helmholtz_eq}
(\Delta-\lambda^2)\preV = -4\pi\rho \, ,
\end{equation}
for a charge density $\rho$ in a periodic domain. 
This equation is frequently referred to as the modified Helmholtz equation or the Yukawa equation. The latter name derives from the  Yukawa potential~\cite{yukawa1935interaction}, $\preV\propto\exp{(-\lambda r)}/r$, in nuclear physics, which is the underlying free-space Green function of \eqref{mod_Helmholtz_eq}. 
In the field of condensed matter, \textit{e.g.}\ in physics, chemistry, and biology,  the Yukawa potential is also known as the screened Coulomb potential. 
It typically emerges in cases when a many-body system of charged particles is treated in terms of an effective single-particle theory applying a mean-field approximation.
Then the many particles contribute to an effective screening of a  Coulomb interaction generated by the single, representative charged particle when treated in linear response theory. 

The relation between the bare Coulomb potential on the one hand and the screened Coulomb potential or the induced screening charge on the other hand is referred to as the dielectric constant or susceptibility, respectively.
Depending on the context, such relations appear in the  Debye–H\"uckel theory~\cite{Debye1923} in the form of a linearization of the Poisson–Boltzmann equation, where the Poisson equation describing the electrostatics of charged particles is a function of the charge density distribution obeying a Boltzmann statistics. Another example is the Thomas-Fermi model~\cite{Thomas1927, Fermi1927} of the dielectric constant in metals, which describes  the screening potential due to the linearized change of the electron distribution described by the Fermi-Dirac distribution with respect  to the spatial variations of the  electrostatic potential.  In these theories, the constant $\lambda$ represents the inverse of a typical length scale over which an individual charged particle exerts a notable effect.

The Thomas-Fermi theory can be regarded as a precursor of the density functional theory~\cite{hohenberg1964inhomogeneous} (DFT). 
The latter is the most important theory and methodology for the modeling and simulation of material properties of a crystalline solid based on the quantum mechanical treatment of many electron systems. 
In addition, the Thomas-Fermi theory provides a rough but fast approximation of the common density functionals, which relate the electron density to the effective Kohn-Sham potentials~\cite{kohn1965density}. 
Such a scheme makes the solution of equation \eqref{mod_Helmholtz_eq} particularly valuable. 
For instance it could be used  to obtain a good starting potential for the iterative solution of the Schr\"odinger-like Kohn-Sham equations, where the nuclear charge is included in $\rho$. 
Other examples are the attainment of an efficient approximate solution of the dielectric function, or the implementation of a hybrid functional~\cite{tran2011screened} to DFT using the Yukawa screening of the Hartree-Fock exchange. 
In both cases the charge density $\rho$ in~\eqref{mod_Helmholtz_eq} is replaced by an overlap charge density~\cite{Massidda:93} obtained as a product of wave functions associated with different quantum numbers. 

Although  most electronic structure methods implementing DFT applied to solid-state materials systems make explicit use of the underlying periodicity of the crystalline lattice, a straightforward solution of \eqref{mod_Helmholtz_eq} using Fourier transformation techniques is in general not possible due to the strongly oscillating charge density close to the nuclei. This problem is well discussed for the solution of the Poisson equation, $\Delta V = -4\pi\rho$, a limit of the modified Helmholtz equation for $\lambda=0$. 

In a seminal work, Weinert~\cite{weinert1981solution} proposed an elegant and numerically efficient  solution of the Poisson equation for periodic charges and corresponding electrostatic potentials without shape approximation. Weinert's solution, to which we refer here as Weinert's pseudo-charge method, is implemented (in several variants) in most full-potential all-electron DFT methods, such as the augmented spherical wave (ASW) method~\cite{eyert2013planewave}, the Korringa-Kohn-Rostoker Green function (KKR-GF) method~\cite{drittler1991thesis}, and the full-potential linearized augmented planewave (FLAPW)  method~\cite{wimmer1981fullpotential}, just to name a few. 

Typical to these all-electron DFT-methods is the domain decomposition into atomic spheres around the atoms and an interstitial region in-between. Weinert's pseudo-charge method is based on the observation that the relation between the charge density inside a sphere and its multipole expansion outside the sphere is not unique. A smooth Fourier transformable  pseudo-charge density with the same multipole moments as the true density is constructed. The latter provides the true potential through Fourier transformation of the Poisson equation and a subsequent solution of a Dirichlet boundary value problem on the sphere boundary. 

In this article, we extend Weinert's pseudo-charge method to the modified Helmholtz equation \eqref{mod_Helmholtz_eq} for values of $\lambda > 0$, and for general periodic charge densities without shape approximation. We formulate our new method for general charge densities, including continuous charge densities as for electron densities, discrete charge densities as for nuclear charges or more abstract densities that arise of products of wave functions. Such an approach is consistent with the real-space representation of the charge density and potential in all-electron methods. 

As a matter of choice, and motivated by the original work of Weinert~\cite{weinert1981solution}, we demonstrate this extension explicitly for the FLAPW method~\cite{wimmer1981fullpotential} as implemented in the \textsc{FLEUR} code~\cite{fleur}.
We provide a complete derivation of the modified multipole expansion using a Green function method, and the derivation of the interstitial charge density's modified multipole moments in the atomic spheres, using some Bessel function integration properties, which yields the coefficients of the pseudo-charge density. 
We also discuss the convergence of the Fourier series of the pseudo-charge density.  
We point out which are the algorithmic changes required to extend the solution of the Poisson equation to a solution of \eqref{mod_Helmholtz_eq}, which can then be straightforwardly transferred  to other all-electron full-potential band structure methods.

This paper is organized as follows: In Sect.~\ref{sec:WpschargeM}, we introduce the muffin-tin and interstitial region typical of the FLAPW method and the corresponding domain decomposition for the charge density and the potential.
We summarize the main statements of Weinert's pseudo-charge method and give a definition of the pseudo-charge density.  
Since we know the true charge density inside the muffin-tin spheres and with the assumption that we would know the interstitial potential, we construct in Sect~\ref{sec:MT-YP} the Yukawa potential inside the sphere by solving the Dirichlet boundary value problem. 
We develop two radial Green functions that are products of two linearly independent fundamental set solutions of the homogeneous radial modified Helmholtz equation. 
These Green functions are set apart by the boundary conditions they fulfill either at the muffin-tin sphere or in free-space. 
In Sect.~\ref{ssec:MMP-expansion}, the radial free-space Green function is used to define the modified multipole expansion of the Yukawa potential. 
In Sect.~\ref{sec:IYP}, we construct a pseudo-charge density in reciprocal space consistent with the modified Helmholtz equation by making use of the definition of the modified multiple moments put forward in Sect.~\ref{ssec:MMP-expansion}. 
We obtain the Yukawa potential for the interstitial region by solving the modified Helmholtz equation in Fourier space for the pseudo-charge density -- the solution is a simple algebraic expression. 
This is followed by an analysis of the convergence properties of the Fourier series of the pseudo-charge density. 
The entire algorithm that solves the modified Helmholtz equation is summarized in \ref{sec:Algorithm} together with the minimal modifications necessary to change Weinert's original algorithm. 
The conclusions and the outlook are presented in Sect.~\ref{sec:conclusion}.

%% file: MainMath.tex
\section{Yukawa Potential for a Muffin-Tin Decomposition of a 3D-periodic Domain}\label{sec:bulk}

\subsection{Weinert's Pseudo-Charge Method}
\label{sec:WpschargeM}

In order to deal with the $1/r$ singularities of the Coulomb potential due to the point-like charge of the nucleus and the associated rapid oscillations of the charge density in the vicinity of the singularity, in all-electron electronic structure methods the space is typically partitioned into muffin-tin spheres $\sphere$ of radius $R_\alpha$ centered around the atoms $\alpha$ -- the union of those is called the muffin-tin (MT) region -- and the interstitial region (I) between the atoms. In FLAPW both charge densities 
\begin{equation}\label{dual_representation} 
\rho(\rr) = \begin{cases} \sum_{\qq} \rhoI(\qq)\,\pw & \rr\in\text{I} \\ \sum_{L} \rho^\alpha_{L}(r_\alpha)\, \Ylm(\rha) & \rr=\ta+\ra\in\sphere \end{cases}
\end{equation}
and potentials
\begin{equation}\label{dual_rep_pot}
V(\rr) = \begin{cases} \sum_{\qq} V^\text{I}(\qq)\,\pw & \rr\in\text{I} \\ \sum_{L} V^\alpha_{L}(r_\alpha)\, \Ylm(\rha) & \rr=\ta+\ra\in\sphere \end{cases} 
\end{equation}
are represented in plane waves $\pw$, where $\qq$ defines the reciprocal lattice vector dual to the lattice vectors defining the periodic domain, and in spherical harmonics, $\Ylm$, of degree $\ell$ and order $m$, where $L$ is defined as $L:=(\ell,m)$.  $r_\alpha\le R_\alpha$ is the length
of the vector $\ra=\rr-\ta$, measured from the center of the atom $\alpha$ placed at position $\ta$ in the periodic domain, with  \mbox{$\rha=\frac{\rr-\ta}{|\rr-\ta|}$} its unit vector. 
The precision of the representation is determined by the cut-off parameters $K_\mathrm{max}$ for the wave vectors, $\qq$, with length $K\le K_\mathrm{max}$, and $\ell_\mathrm{max}$ for the degree $\ell$ in the angular-momentum expansion.
$\ell_\mathrm{max}$ sets also a  natural cut-off of the angular-momentum expansions of all other charge densities or multipole moments throughout the paper.

Weinert's pseudo-charge method for the Poisson equation is based on the  crucial observation that several charge densities $\rho$ inside a sphere $\sphere$ can generate the same multipole moments 
\begin{equation}
\qlm[\rho] = \int_{\sphereo}\rho(\ra+\ta) r_\alpha^{\ell} \Ylm^\ast(\rha)\diff\ra
\end{equation}
and thus, the same potential
\begin{equation}
V^{\text{I}}[\rho](\rr) = \sum_{L} \frac{4\pi}{2\ell+1}\qlm[\rho] \frac{1}{r_\alpha^{\ell +1}}\Ylm(\rha)
\end{equation}
outside the sphere. Here, $\Ylm^\ast$ denotes the complex conjugate of $\Ylm$.
The pseudo-charge density, $\pseudo$, defined by Weinert in Ref.~\onlinecite{weinert1981solution} is such a charge density. It fulfills the following three conditions:
\begin{itemize}
  \item It has the same multipole moments $\qlm[\rhot]=\qlm[\rho]$ in every sphere $\sphere$.
  \item It is equal to the true charge density $\rhoI$ in the interstitial region.   
  \item It has a fast convergent Fourier expansion.
\end{itemize}
The Fourier components of $V^\text{I}$ are then simply \begin{equation}
V^\text{I}(\qq)=\frac{4\pi}{K^2}\pseudo(\qq) \quad\text{for } \qq\neq\mathbf{0}\,,
\end{equation}
while $V^\text{I}(\mathbf{0})$ will be set to a constant.
Once the interstitial potential $V^\text{I}$ has been calculated,
the muffin-tin potential can be obtained by solving the Dirichlet boundary value problem on the sphere
\begin{align}\label{3D_MTPot}
V^\alpha(\ra+\ta) &= V_\text{S}^\alpha(\ra+\ta) + V_\text{B}^\alpha(\ra+\ta)\nonumber\\
 &= \int_{\sphereo}G(\ra,\rap)\rho(\rap+\ta)\diff\rap \\
&- \frac{R_\alpha^2}{4\pi}\int_{\partial\sphereo}V^\text{I}(\rap+\ta)\frac{\partial G}{\partial n'}(\ra,\rap)\diff\omega',\nonumber
\end{align}
where $G$ is a Green function associated with the solution of the Poisson equation,  $\diff\omega=\sin\theta\diff\theta\diff\phi$ denotes the solid angle element and $\rr=\ra+\ta\in\sphere$. 
Although the Green function depends on the muffin-tin radius $R_\alpha$, for simplicity we drop the index $\alpha$ in the Green function and in related quantities. 
The muffin-tin potential $V^\alpha$ is fed by two terms, a source term $V^\alpha_\text{S}$ due to the charge density distribution inside the sphere and a boundary term $V^\alpha_\text{B}$ due to the interstitial potential at the boundary of the sphere.
The Fourier coefficients of the pseudo-charge density basically have the form
\begin{equation}\label{rhot_Fourier}
\rhot(\qq)=\rhoI(\qq)+\sum_\alpha\rhoab(\qq)\,, 
\end{equation}
where $\rhoab$ is a Fourier transformable pseudo-charge density inside the muffin-tin sphere. 
The idea behind this is the following:
if the domain of definition of $\rhoI$ is formally expanded to the full space, \textit{i.e.}\ including the muffin-tin spheres, such that 
\begin{equation}
    \rhoI(\rr) = \sum_{\qq} \rhoI(\qq)\pw
\end{equation}
can also be evaluated for $\rr\in\cup_\alpha\sphere$, then the true charge $\rho$ can also be written as
\begin{equation}
\rho=\rhoI+\sum_\alpha\breve{\rho}^\alpha\,,
\end{equation}
where 
\begin{equation}\label{eq:def_charge-in-MT}
\breve{\rho}^\alpha:=\begin{cases}0&\text{in I}\\\rhoa-\rhoI&\text{in }\sphere\end{cases}
\end{equation}
are charge densities localized in the atomic spheres $\sphere$.
If these localized densities are now substituted by other localized densities $\rhoab$, then the charge density is still correct in I. 

We now continue with the derivation of the muffin-tin and interstitial potentials for $\lambda>0$.

\input{MuffinTinPotential}
\input{InterstitialPotential}
\input{AlgorithmBulk}

%% file: MuffinTinPotential.tex
\subsection{Muffin-Tin Yukawa Potential}\label{sec:MT-YP}

Assume the interstitial potential is obtained, then the Green function method is used to determine the screened Coulomb potential $\preV^\alpha$ inside the muffin-tin sphere $\sphere$ with centre $\ta$ through the solution of the boundary value problem
\begin{align}
(\Delta-\lambda^2)\preV^\alpha &= -4\pi\rho \hphantom{\preV^\text{I}}\quad\text{in }\sphere\label{MT_eq}\\
\preV^\alpha &= \preV^\text{I} \hphantom{-4\pi\rho}\quad\text{on }\partial\sphere \, .
\end{align}
The solution is divided into three steps, which we describe in the following. 
The derivation focuses on the construction of the Green function and its application.
For simplicity, we leave out the index $\lambda$ in Green functions and potentials in this subsection. 
The solution is given in terms of radial functions $V_{L}^\alpha(r_\alpha)$, the expansion coefficients  to the spherical harmonics expansion inside the sphere (see \eqref{dual_rep_pot}).

\noindent
\textit{Step 1.} We solve the homogeneous modified Helmholtz equation, $(\Delta -\lambda^2)U=0$, in spherical coordinates. 
Following the solution~\cite{abramowitz2008handbook} of the Laplace equation,  $\Delta \psi=0$, in spherical coordinates $(r,\rh)$, the homogeneous potential can be factorized into products of radial functions $u_\ell(r)$ and  angular functions $\Ylm(\rh)$, $U(r,\rh)=\sum_L u_\ell(r)\Ylm(\rh)$. 
The term $-\lambda^2 U$ in the homogeneous modified Helmholtz equation only has an effect on the radial solution. The spherical harmonics, $\Ylm$, are the eigensolutions of the angular part of the Laplace equation with eigenvalues $\ell(\ell+1)$. 
The radial part of the homogeneous modified Helmholtz equation is known as the modified spherical Bessel differential equation~\cite{abramowitz2008handbook,arfken2013mathematical}, 
\begin{equation}\label{mod_sph_bessel_diff_eq}
\frac{\diff^2 u_{\ell}(r)}{\diff r^2} + \frac{2}{r}\frac{\diff u_{\ell}(r)}{\diff r} - \left( \frac{\ell(\ell+1)}{r^2} + \lambda^2 \right ) u_{\ell}(r) = 0,
\end{equation}
and its fundamental set of solutions
are for each $\ell$ the two modified spherical Bessel functions~\cite{abramowitz2008handbook,arfken2013mathematical} $i_{\ell}(\lambda r)$ and $k_{\ell}(\lambda r)$, the first of which is the regular solution well-defined at the origin, but grows fast with growing radius $r$ and the second is the irregular solution that goes to infinity for $r\to 0$.
To realize the proper boundary condition for the radial Green functions two conditions have to be fulfilled:
The first solution, $u_{\ell 1}$, must be finite at $r=0$. 
We conclude that 
\begin{equation}\label{condition_u1}
u_{\ell 1}(r)=i_{\ell}(\lambda r)\,,
\end{equation}
since $k_{\ell}(\lambda r)\rightarrow\infty$ for $r\to0$. The second solution, $u_{\ell 2}$, must be 0 at $r=R_\alpha$. 
This is achieved by a linear combination of the two modified spherical Bessel functions,
\begin{equation}\label{condition_u2}
u_{\ell 2}(r) = k_{\ell}(\lambda r) - i_{\ell}(\lambda r) \frac{k_{\ell}(\lambda R_\alpha)}{i_{\ell}(\lambda R_\alpha)}\,.
\end{equation}

\noindent
\textit{Step 2. } A function $G_l\in C_0([0,R_\alpha]\times[0,R_\alpha])\cap C_2([0,R_\alpha]\times[0,R_\alpha]\backslash\{(r,r)\vert r\in[0,R_\alpha]\})$ is called a radial Green function, if it is the solution to
\begin{equation}\label{def_sph_green}
(\Delta_r-\lambda^2) G_{\ell}(r,r') = -\frac{4\pi}{r^2}\delta(r-r') 
\end{equation}
subject to the Dirichlet boundary condition
\begin{equation}
G_{\ell}(r,r')=0\,, \qquad\text{if}\,\,\> r=R_\alpha \quad\text{or}\quad r'=R_\alpha\, ,
\end{equation}
where $(\Delta_r-\lambda^2)$ is the linear radial differential operator in~\eqref{mod_sph_bessel_diff_eq} and $\delta$ denotes the radial Dirac delta function, for which
\begin{equation}
\int_a^b\delta(r-r')f(r')\diff r' = \begin{cases} f(r), & \text{if } r\in[a,b]\\ 0, & \text{otherwise. }\end{cases}
\end{equation}
The radial Green function takes the form of the product of the two linearly independent solutions 
with the proper boundary conditions,
\begin{equation}\label{sph_green}
G_{\ell}(r,r') = C u_{\ell 1}(r_<)u_{\ell 2}(r_>)\,,
\end{equation}
where $r_<=\min(r,r')$, $r_>=\max(r,r')$ and 
\begin{equation}
C=-\frac{4\pi}{r'^2}W^{-1}\left[u_{\ell 1}(r'),u_{\ell 2}(r')\right]\,.
\end{equation}
Since the Wronskian $W$ is linear, the addition of $cu_{\ell 1}$ ($c=$const) onto $k_{\ell}(\lambda r)$ to suffice the boundary condition $u_{\ell 2}(R_\alpha)=0$ has no influence on the Wronskian: $W(u_{\ell 1},cu_{\ell 1})=0$, with $W(i_{\ell},k_{\ell})$ remaining.
To calculate the Wronskian of $i_{\ell}(\lambda r')$ and $k_{\ell}(\lambda r')$ one can either Taylor expand the two functions or simply take the limiting values for $r\to0$ or the asymptotic values for $r\to\infty$ to find
\begin{align}\label{wronskian}
W\left[u_{\ell 1}(r^\prime),u_{\ell 2}(r^\prime)\right] &= u_{\ell 1}(r^\prime)\frac{\diff u_{\ell 2}(r^\prime)}{\diff r^\prime}-\frac{\diff u_{\ell 1}(r^\prime)}{\diff r^\prime}u_{\ell 2}(r^\prime) \nonumber\\&= -\frac{1}{\lambda r^{\prime2}}\,.
\end{align}
Therefore, $C  = 4\pi\lambda$, and the radial Green function finally reads  
\begin{align}\label{sph_greenfct}
G_{\ell}(r,r')&= 4\pi\lambda\,i_{\ell }(\lambda r_<)k_{\ell }(\lambda r_>)\nonumber \\
& - 4\pi\lambda\,i_{\ell }(\lambda r)  i_{\ell}(\lambda r^\prime) \frac{k_{\ell}(\lambda R_\alpha)}{i_{\ell}(\lambda R_\alpha)}\, .
\end{align}

\noindent
\textit{Step 3.} Considering the standard expression of Dirac's delta function separated according to the radial and angular coordinates 
\begin{align}
\delta(\rr-\rp) &=\frac{1}{r^2}\delta(r-r^\prime)\delta(\rh-\rhp) \nonumber \\
&=\frac{1}{r^2}\delta(r-r^\prime)\sum_L\Ylm^*(\rhp)\Ylm(\rh)\,,
\end{align}
the three-dimensional (3D) Green function $G(\rr,\rp)$ solving 
\begin{align}
(\Delta-\lambda^2)G(\rr,\rp) &= -4\pi\delta(\rr-\rp) \hphantom{0i}\text{in }\sphereo\,,\\
G(\rr,\rp) &= 0 \hphantom{-4\pi\delta(\rr-\rp)}\text{on }\partial\sphereo\,,
\end{align}
is expanded in the form 
\begin{equation}\label{3D_Green}
G(\rr,\rp) =\sum_{L}G_{\ell}(r,r')\Ylm^\ast(\rhp)\Ylm(\rh)
\end{equation}
and the solution to the inhomogeneous equation~\eqref{MT_eq} is given by \eqref{3D_MTPot}.
For the derivation of the 3D Green function and the 3D inhomogeneous solution $V^\alpha=V_\text{S}^\alpha+V_\text{B}^\alpha$~\eqref{3D_MTPot} we refer the reader to Ref.~\onlinecite{jackson1999classical}.
Both the integral over the sphere $\sphereo$ and the boundary integral simplify by exploiting the orthonormality relation of the spherical harmonics, $\int_{\partial B_1(\nullvec)}\Ylm^\ast(\rh) Y_{L^\prime}(\rh)\diff\omega=\delta_{LL^\prime}$. 
The integral over the sphere $\sphereo$ provides the source contribution to the muffin-tin potential
\begin{align}
V_\text{S}^\alpha(\ra+\ta)&=\int_{\sphereo}G(\ra,\rap)\rho(\rap+\ta)\diff\rap \nonumber\\
&= \sum_{L}\left[\int_0^{R_\alpha} G_{\ell}(r_\alpha,r_\alpha^\prime)\rho_{L}^\alpha(r_\alpha^\prime)r_\alpha^{\prime 2}\diff r_\alpha^\prime\right] \Ylm(\rha)\,.
\end{align}

In order to obtain the boundary contribution to the muffin-tin potential, we evaluate the boundary integral in~\eqref{3D_MTPot} by expanding the interstitial potential $V^\text{I}(\qq)$ (see Sect.~\ref{sec:IYP}) on the sphere boundaries $\partial\sphere\ni\rp$ in spherical coordinates
\begin{align}
V^\text{I}(\rap+\ta) 
&= \sum_{\qq} V^\text{I}(\qq) \me^{\iu\qq\cdot\ta} \me^{\iu\qq\cdot\rap} \nonumber\\
&= \sum_{L} \VIR \Ylm(\rhap)
\end{align}
using the plane-wave expansion 
\begin{equation}\label{pw_expansion}
\me^{\iu\qq\cdot\rr} = \sum_{L}4\pi\iu^{\ell}j_{\ell}(Kr)\Ylm^\ast(\hat{\qq})\Ylm(\rh)\,,
\end{equation}
where 
\begin{equation}\label{spherical_harmonics_boundary_term}
\VIR = 4\pi\iu^{\ell} \sum_{\qq} V^\text{I}(\qq) \me^{\iu\qq\cdot\ta}j_{\ell}(KR_\alpha) \Ylm^\ast(\hat{\qq})\,.
\end{equation}
Furthermore, the normal derivative of $G$ on the sphere boundary is
\begin{align}
\frac{\partial G}{\partial n^\prime}(\ra,\rap) &= \frac{\partial G(\ra,\rap)}{\partial r_\alpha^\prime}\bigg\vert_{r_\alpha^\prime=R_\alpha} \\
&= \sum_{L} \frac{\partial G_{\ell}(r_\alpha,r_\alpha^\prime)}{\partial r_\alpha^\prime}\bigg\vert_{r_\alpha^\prime=R_\alpha}\Ylm^\ast(\rhap) \Ylm(\rha)\,.\nonumber
\end{align}
Since $r_\alpha<R_\alpha=r_\alpha^\prime$ and since $G_{\ell}$ takes the form~\eqref{sph_green}, we obtain 
\begin{equation}
\frac{\partial G_{\ell}(r_\alpha,r_\alpha^\prime)}{\partial r_\alpha^\prime}\bigg\vert_{r_\alpha^\prime=R_\alpha} = 4\pi\lambda\, u_{\ell 1}(r_\alpha) u_{\ell 2}^\prime(R_\alpha)\,.
\end{equation}
We recall that $u_{\ell 2}(R_\alpha)=0$ and  reuse~\eqref{wronskian} to obtain
\begin{equation}
u_{\ell 2}^\prime(R_\alpha) = -\frac{1}{\lambda R_\alpha^2 i_{\ell}(\lambda R_\alpha)}\,,
\end{equation}
yielding
\begin{equation}
\frac{\partial G_{\ell}(r_\alpha,r_\alpha^\prime)}{\partial r_\alpha^\prime}\bigg\vert_{r_\alpha^\prime=R_\alpha} = -\frac{4\pi}{R_\alpha^2} \frac{i_{\ell}(\lambda r_\alpha)}{i_{\ell}(\lambda R_\alpha)}\,.
\end{equation}
With this and the knowledge of the interstitial potential at the sphere boundary from~\eqref{spherical_harmonics_boundary_term}, the boundary contribution to the muffin-tin potential becomes
\begin{align}
V^\alpha_\text{B}(\ra+\ta)&=-\frac{R_\alpha^2}{4\pi}\int_{\partial\sphereo}V^\text{I}(\rap+\ta)\frac{\partial G}{\partial n'}(\ra,\rap)\diff\omega' \nonumber\\
&= \sum_{L} \VIR \frac{i_{\ell}(\lambda r_\alpha)}{i_{\ell}(\lambda R_\alpha)} \Ylm(\rha)
\end{align}
and the radial part of the spherical harmonics expansion of the total potential, $V^\alpha_\text{S}+V^\alpha_\text{B}$, in the sphere $\sphereo$ becomes 
\begin{equation}\label{potential_sphere_algorithm}
\begin{split}
V_{L}^\alpha(r_\alpha) &= \int_0^{R_\alpha} G_{\ell}^\alpha(r_\alpha,r_\alpha^\prime) \rho_{L}^\alpha(r_\alpha^\prime) r_\alpha^{\prime 2} \diff r_\alpha^\prime \\
&+ \VIR \frac{i_{\ell}(\lambda r_\alpha)}{i_{\ell}(\lambda R_\alpha)}\,.
\end{split}
\end{equation}
Due to the kink of $G_{\ell}$ at $r_\alpha=r'_\alpha$, for practical calculations the integral is split in a part where $r'_\alpha<r_\alpha$, a part where $r'_\alpha>r_\alpha$ and a third part where the integrand is symmetric in $r_\alpha$ and $r'_\alpha$:
\begin{align}\label{computable_MTPot}
V_{L}^\alpha(r_\alpha) &= 4\pi\lambda \left( \left[\int_0^{r_\alpha} \rho_{L}^\alpha(r'_\alpha)i_{\ell}(\lambda r'_\alpha) r_\alpha^{\prime 2} \diff r'_\alpha\right] k_{\ell}(\lambda r_\alpha) \right.\nonumber\\
&+ \left. \left[\int_{r_\alpha}^{R_\alpha} \rho_{L}^\alpha(r'_\alpha)k_{\ell}(\lambda r'_\alpha) r_\alpha^{\prime 2} \diff r'_\alpha\right] i_{\ell}(\lambda r_\alpha) \right. \nonumber\\
&-\left.  \left[\int_0^{R_\alpha} \rho_{L}^\alpha(r'_\alpha)i_{\ell}(\lambda r'_\alpha) r_\alpha^{\prime 2} \diff r'_\alpha\right] i_{\ell}(\lambda r_\alpha)\frac{k_{\ell}(\lambda R_\alpha)}{i_{\ell}(\lambda R_\alpha)} \right) \nonumber\\
&+ \VIR \frac{i_{\ell}(\lambda r_\alpha)}{i_{\ell}(\lambda R_\alpha)}\,. 
\end{align}

\subsubsection{Modified Multipole Expansion}\label{ssec:MMP-expansion}

In the same way we obtain the radial representation of the free-space Green function, well-known as the Yukawa potential for a Dirac test charge at $\rp$, 
\begin{equation}\label{eq:free-space-GF}
    \frac{\me^{-\lambda |\rr-\rp|}}{|\rr-\rp|} = 4\pi\lambda\sum_L  i_\ell(\lambda r_<) k_\ell(\lambda r_>) \Ylm^\ast(\rhp) \Ylm(\rh)\, .
\end{equation}
The modified spherical Bessel function $k_\ell$ already contains the proper boundary condition for $r\rightarrow\infty$. 
A charge density $\breve{\rho}^\alpha$ localized in a sphere $\sphere$ embedded in free space, produces a Yukawa potential outside the sphere, \textit{i.e.}\ $\rr=\ra+\ta\notin\sphere$, 
\begin{equation}\label{eq:Yukawa_x_charge}
V^{\text{I}}[\breve{\rho}^\alpha](\rr)= \int_{\mathbb{R}^3} G(\ra,\rap)\, \breve{\rho}^\alpha(\rap+\ta)\diff\rap\, ,
\end{equation}
which can be expressed analogously to the Coulomb potential in terms of the modified multipole expansion
\begin{equation}\label{eq:V_alpha-I}
V^{\text{I}}[\breve{\rho}^\alpha](\rr) = \sum_{L} \frac{4\pi\lambda^{\ell +1}}{(2\ell+1)!!}\qlm[\breve{\rho}^\alpha]\, k_{\ell}(\lambda r_\alpha)\,\Ylm(\rha)\,,
\end{equation}
with the modified multipole moments
\begin{equation}\label{mod_mpm}
\begin{split}
\qlm[\breve{\rho}^\alpha] = \frac{(2\ell+1)!!}{\lambda^\ell}\int_{\sphereo}&\breve{\rho}^\alpha(\ra+\ta)\\ & i_{\ell}(\lambda r_\alpha)\Ylm^\ast(\rha)\diff\ra\, .
\end{split}
\end{equation}
An analogous definition holds true for the modified multiple moments $\qlm[\rhoa]$ of the true charge $\rhoa$ in the sphere.
With the standard expansion of the charge density inside the sphere into spherical harmonics~\eqref{dual_rep_pot}, $\rhoa(\ra+\ta) =\sum_L \rhoa_L(r_\alpha)\,\Ylm(\rha)$ and the application of their orthonormality relation, the calculation of the modified multipole moments inside the muffin-tin spheres is straightforward and results in
\begin{equation}\label{MTqlm}
\qlm[\rhoa] = \frac{(2\ell+1)!!}{\lambda^{\ell}} \int_0^{R_\alpha} \rhoa_{L}(r_\alpha)\, i_{\ell}(\lambda r_\alpha)\, r_\alpha^{2} \diff r_\alpha\, .
\end{equation}
The summation of $V^{\text{I}}[\breve{\rho}^\alpha]$ over all spheres $\alpha$ finally provides the contributions of the charges of all spheres to the interstitial potential. Equation \eqref{3D_MTPot} reduces to \eqref{eq:Yukawa_x_charge}, since for the free-space Green function \eqref{eq:free-space-GF} the boundary value term disappears for $r\rightarrow \infty$.

%% file: InterstitialPotential.tex
\subsection{Interstitial Yukawa Potential}\label{sec:IYP}

Suppose we had found a Fourier transformable pseudo-charge density $\rhoab$ inside the sphere consistent with the Yukawa potential produced outside the sphere, with coefficients $\rhoab(\qq)$, and the Fourier series would converge rapidly throughout the periodic domain. 
Then we can find the Fourier coefficients of the pseudo-charge density, $\rhot(\qq)$, by \eqref{rhot_Fourier} and the solution of the modified Helmholtz equation \eqref{mod_Helmholtz_eq} through Fourier transformation yields an algebraic equation from which we calculate the interstitial Yukawa potential,
\begin{equation}\label{preVI_Fourier}
\preV^\text{I}(\qq) = \frac{4\pi}{K^2+\lambda^2}\rhot(\qq)\, .
\end{equation} 
In the previous subsection~\ref{sec:MT-YP} the interstitial Yukawa potential is used as boundary values for the Yukawa potential in the atomic spheres. 

This subsection is concerned with the construction of the Fourier transformable pseudo-charge density $\rhoab$ that replaces the true local charge density $\breve{\rho}^\alpha$ (see \eqref{eq:def_charge-in-MT}) inside the muffin-tin sphere such that the Yukawa potential in the interstitial region,  $V_\lambda^{\text{I}}[\breve{\rho}^\alpha]=V_\lambda^{\text{I}}[{\rho}^\alpha]- V_\lambda^{\text{I}}[\rhoI]$, due to the true charge density inside the sphere, is equal to the Yukawa potential in the interstitial region produced by the \textit{a priori} unknown pseudo-charge density, $V_\lambda^{\text{I}}[\rhoab]=V_\lambda^{\text{I}}[\breve{\rho}^\alpha]$.
From~\eqref{eq:V_alpha-I} we conclude that this is fulfilled if the modified multipole moments~\eqref{mod_mpm} of both charge densities, $\qlm[\rholoc]=\qlm[\rhoa]-\qlm[\rhoI]$ and $\qlm[\rhoab]$, are equal. 
The modified multiple moments $\qlm[\rho^\alpha]$  are already known through \eqref{MTqlm}. Next we determine the modified multiple moments $\qlm[\rhoI]$ of the interstitial charge extended into the muffin-tin spheres and then construct the pseudo-charge density. 

\subsubsection{Modified Multiple Moments of Interstitial Charge Density Extended into Sphere} \label{ssec:MMMICh}
The determination of the modified multiple moments of the interstitial charge density is in principle the same as in Ref.~\onlinecite{weinert1981solution}, but since the modified multipole moments are different from the known multipole moments for the Coulomb potential, we go through this step of deriving $\qlm[\rhoI]$ in detail.
We write $\rhoI$ relative to the sphere centre $\ta$ and employ the Rayleigh expansion~\eqref{pw_expansion} to $\me^{\iu\qq\cdot\ra}$, which yields
\begin{align}
\rhoI(\rr) &=\sum_\qq \rhoI(\qq)\me^{\iu\qq\cdot\ra}\me^{\iu\qq\cdot\ta}\\ &=\sum_\qq \rhoI(\qq)\me^{\iu\qq\cdot\ta}\sum_{L}4\pi\iu^{\ell} j_{\ell}(Kr_\alpha)\Ylm^\ast(\hat{\qq})\Ylm(\rha)\,.\nonumber
\end{align}
The modified multipole moments of $\rhoI$ in the sphere $\sphere$, defined analogously to~\eqref{mod_mpm}, are
\begin{align}
\qlm[\rhoI] &= \frac{(2\ell+1)!!}{\lambda^{\ell}} \int_{\sphereo} \Ylm^\ast(\rha) i_{\ell}(\lambda r_\alpha) \rhoI(\ra+\ta)\diff\ra \nonumber\\
&= \frac{(2\ell+1)!!}{\lambda^{\ell}} 4\pi \iu^{\ell} \sum_\qq\rhoI(\qq)\me^{\iu\qq\cdot\ta}\Ylm^\ast(\hat{\qq}) \\
&\phantom{\frac{(2\ell+1)!!}{\lambda^{\ell}} 4\pi \iu^{\ell} \sum_\qq\rhoI}\int_0^{R_\alpha} i_{\ell}(\lambda r_\alpha) j_{\ell}(Kr_\alpha) r_\alpha^2 \diff r_\alpha\,. \nonumber
\end{align}
For $\qq \ne 0$, the latter integral becomes
\begin{align}
&\int_0^{R_\alpha} i_{\ell}(\lambda r_\alpha) j_{\ell}(K r_\alpha) r_\alpha^2 \diff r_\alpha \nonumber\\
&= \frac{R_\alpha^2}{K^2+\lambda^2}\left(K i_{\ell}(\lambda R_\alpha)j_{\ell +1}(K R_\alpha)+\lambda i_{\ell +1}(\lambda R_\alpha)j_{\ell}(K R_\alpha)\right)\nonumber\\
&= \frac{R_\alpha^2}{K^2+\lambda^2}\left(\lambda i_{\ell -1}(\lambda R_\alpha)j_{\ell}(K R_\alpha)- K i_{\ell}(\lambda R_\alpha)j_{\ell -1}(K R_\alpha)\right)\, .
\end{align}
If $\qq = 0$, observe that $j_{\ell}(0)=\delta_{\ell 0}$ (Kronecker $\delta$) and so the integral becomes 
\begin{equation}
\delta_{\ell 0} \int_0^{R_\alpha} i_0(\lambda r_\alpha) r_\alpha^2 \diff r_\alpha = R_\alpha^3\frac{ i_1(\lambda R_\alpha)}{\lambda R_\alpha} \delta_{\ell 0}\, .
\end{equation}
Both equations above can be derived by partial integration in two different ways and applying the identities in Ref.~\onlinecite{arfken2013mathematical}, 
\begin{equation}
\frac{\diff}{\diff r}\left(r^{-\ell}f_{\ell}(r)\right) = \pm r^{-\ell}f_{\ell +1}(r)
\end{equation}
for $f_{\ell}=i_{\ell}$ with the plus sign and for $f_{\ell}=j_{\ell}$ with the minus sign, and
\begin{equation}\label{bessel_identity}
\frac{\diff}{\diff r}\left(r^{\ell +2}f_{\ell +1}(r)\right) = r^{\ell +2}f_{\ell}(r)
\end{equation}
for both $f_{\ell}=i_{\ell}$ and $f_{\ell}=j_{\ell}$.
In conclusion, this yields
\begin{flalign}\label{Iqlm}
\qlm[\rhoI] &= \delta_{\ell 0}\sqrt{4\pi}\frac{R_\alpha^2i_1(\lambda R_\alpha)}{\lambda}\rhoI(\mathbf{0}) \nonumber\\
&+ \sum_{\qq\ne \mathbf{0}}\frac{(2\ell+1)!!}{\lambda^{\ell}} 4\pi\iu^{\ell}\rhoI(\qq)\me^{\iu\qq\cdot\ta}\Ylm^\ast(\hat{\qq}) \frac{R_\alpha^2}{\lambda^2+K^2}\nonumber\\
&\quad(K i_{\ell}(\lambda R_\alpha)j_{\ell +1}(KR_\alpha)+\lambda i_{\ell +1}(\lambda R_\alpha)j_{\ell}(KR_\alpha))\,.
\end{flalign}

\subsubsection{Construction of Pseudo-Charge Density}\label{ssec:cPseudoCD}

We construct the pseudo-charge density by following the Ansatz of Weinert,
\begin{align}\label{ansatz_rhoab}
\rhoab(\ra+\ta) &= \sum_{L}\rhoab_{L}(r_\alpha)\Ylm(\rha)\nonumber\\&=\sum_{L} Q_{Ln}^\alpha \left(\sum_{\eta=0}^n a_\eta r_\alpha^{\nu_\eta}\right) \Ylm(\rha)\,,
\end{align}
in which the radial dependence of the charge density is expressed in terms of a polynomial expansion up to degree $\nu_n$, which depends on atom  $\alpha$ and angular degree $\ell$, and otherwise use spherical harmonics for the angular part---this being the usual representation for charge densities in the muffin-tin region. 
As we will discuss in Sect.~\ref{ssec:fourierconvergence}, it is beneficial to choose
\begin{equation}\label{coefficients_a_eta}
a_{\eta}=(-1)^{n-\eta}R_\alpha^{2(n-\eta)}\binom{n}{\eta}a_n \quad \text{for } \eta=0,\ldots,n
\end{equation}
and $\nu_\eta=\ell +2\eta$, where $n$ is yet to be determined.
As will become apparent later when we derive the coefficients $Q_{Ln}^\alpha$ in~\eqref{coeff_QL}, the coefficient $a_n$ cancels out in any relevant equation, like~\eqref{reduced_rhoab} or~\eqref{rhoab_Fourier}.
With these choices of parameters and with the binomials theorem applied to
\begin{equation}
\sum_{\eta=0}^n (-1)^{n-\eta} \binom{n}{\eta} \left(\frac{r_\alpha}{R_\alpha}\right)^{2\eta} = \left(\left(\frac{r_\alpha}{R_\alpha}\right)^2-1\right)^{n}
\end{equation}
it follows from the ansatz~\eqref{ansatz_rhoab}
\begin{equation}\label{reduced_rhoab}
\rhoab(\ra+\ta) 
= a_n (r_\alpha^2-R_\alpha^2)^n \sum_L Q_{Ln}^\alpha r_\alpha^\ell \Ylm(\rha)\, .
\end{equation}
The Fourier transform of this expression is then given by
\begin{align}\label{rhoab_Fourier}
\rhoab(\qq) &= \frac{1}{|\Omega|} \me^{-\iu\qq\cdot\ta} \int_{\sphereo} \rhoab(\ra+\ta) \me^{-\iu \qq\cdot\ra} \diff\ra \nonumber\\
&=\frac{ 4\pi}{|\Omega|} \me^{-\iu\qq\cdot\ta} \sum_{L} (-i)^{\ell} Q_{Ln}^\alpha A_{\ell n}^\alpha(K) \Ylm(\hat{\qq})\, ,
\end{align}
where  
\begin{equation}\label{coeffdef_Al}
A_{\ell n}^\alpha(K) = a_n \int_0^{R_\alpha}(r_\alpha^2-R_\alpha^2)^n r_\alpha^{\ell+2} j_\ell(Kr_\alpha) \diff r_\alpha\, ,
\end{equation}
$|\Omega|$ is the volume of the periodic domain and we used the Rayleigh expansion~\eqref{pw_expansion} and the orthonormality relation of the spherical harmonics.
With Prop.~\ref{app:1} in Appendix~\ref{app}
$A_{\ell n}^\alpha(K)$ finally reduces to 
\begin{equation}\label{coeff_Al}
A_{\ell n}^\alpha(K) = a_n (-2)^n n! R_\alpha^{\ell +n+2} \frac{j_{\ell +n+1}(KR_\alpha)}{K^{n+1}}\, .
\end{equation}
In the same way we derive the coefficients $Q_{Ln}^\alpha$: 
We insert~\eqref{reduced_rhoab} in the definition of the modified multipole moments~\eqref{mod_mpm} and use~\eqref{bessel_integration_property} for $f_\ell=i_\ell$ and $\kappa=\lambda$ to obtain
\begin{align}
\qlm[\rhoab] 
&= \frac{(2\ell+1)!!}{\lambda^{\ell}} \int_{\sphereo} \rhoab(\ra+\ta) i_\ell(\lambda r_\alpha) \Ylm^\ast(\rha) \diff \ra \\
&= \frac{(2\ell+1)!!}{\lambda^{\ell}} Q_{Ln}^\alpha a_n \int_0^{R_\alpha} (r_\alpha^{2}-R_\alpha^2)^n r_\alpha^{\ell+2} i_\ell(\lambda r_\alpha) \diff r_\alpha \\
&= \frac{(2\ell+1)!!}{\lambda^{\ell}} Q_{Ln}^\alpha a_n (-2)^n n! R_\alpha^{\ell +n+2} \frac{i_{\ell +n+1}(\lambda R_\alpha)}{\lambda^{n+1}}\, ,
\end{align}
and thus,
\begin{equation}\label{coeff_QL}
Q_{Ln}^\alpha = \qlm[\rhoab] \frac{\lambda^{\ell +n+1}}{a_n (-2)^n n! R_\alpha^{\ell +n+2} (2\ell+1)!! i_{\ell +n+1}(\lambda R_\alpha)}\, .
\end{equation}
Since $A_{\ell n}^\alpha(K)$ and $Q_{Ln}^\alpha$ share the term $a_n (-2)^n n! R_\alpha^{\ell +n+2}$, the term cancels out in the product $Q_{Ln}^\alpha A_{\ell n}^\alpha(K)$ entering~\eqref{rhoab_Fourier},
\begin{equation}\label{Yukawa_factor}
Q_{Ln}^\alpha A_{\ell n}^\alpha (K) = \frac{j_{\ell +n+1}(KR_\alpha)}{K^{n+1}(2\ell+1)!!} \frac{\lambda^{\ell +n+1}}{i_{\ell +n+1}(\lambda R_\alpha)} \qlm[\rhoab]\,,
\end{equation}
and setting $\nu=\ell +n+1$ (in accordance with Sect.~\ref{ssec:fourierconvergence}) this leads to
\begin{equation}\label{rhoab_Fourier_final}
\begin{split}
\rhoab(\qq) = \frac{4\pi}{|\Omega|} \me^{-\iu\qq\cdot\ta} \sum_{L} & (-i)^{\ell} \frac{j_\nu(KR_\alpha)}{K^{\nu-\ell}(2\ell+1)!!} \frac{\lambda^{\nu}}{i_{\nu}(\lambda R_\alpha)}\\ & \qlm[\rhoab]\,  \Ylm(\hat{\qq})\, .
\end{split}
\end{equation}
For $\rhoab(\mathbf{0})$ we take the limit 
\begin{equation}
\lim_{K\to0}\frac{j_{\nu}(KR_\alpha)}{K^{\nu-\ell}} = \lim_{K\to0}\frac{K^{\ell} R_\alpha^{\nu}}{(2\nu+1)!!}\,,
\end{equation}
which is only different from 0 for $\ell=0$ and thus yields
\begin{equation}\label{rhoab0}
\rhoab(\mathbf{0}) = \frac{\sqrt{4\pi}}{|\Omega|}\frac{(\lambda R_\alpha)^{\nu}}{(2\nu+1)!!i_{\nu}(\lambda R_\alpha)} q_{00}^\alpha[\rhoab]\, .
\end{equation}
Due to the condition that the pseudo-charge density has the correct modified multipole moments, the modified multipole moments used for the actual computation of the Fourier coefficients are the ones of the true localized charge density $\rholoc$, $\qlm[\rholoc] = \qlm[\rhoa]-\qlm[\rhoI]$, calculated from the modified multipole moments $\qlm[\rhoa]$~\eqref{MTqlm} and $\qlm[\rhoI]$~\eqref{Iqlm} of $\rhoa$ and $\rhoI$ respectively, in the sphere $\sphere$.

\subsubsection{Smoothness of the Pseudo-Charge Density and Convergence of Its Fourier Series}\label{ssec:fourierconvergence}

In Sect.~\ref{ssec:cPseudoCD}, we have set $a_\eta$ by~\eqref{coefficients_a_eta}, $\nu_\eta=\ell +2\eta$ and $\nu=\ell +n+1$ without having determined $n$ yet. 
Here we motivate our choices and finally determine a proper $n$. 

With our choices for $a_\eta$ and $\nu_\eta$ we have eradicated the sum in ansatz~\eqref{ansatz_rhoab} and derived the much simpler form~\eqref{reduced_rhoab}.
The function 
\begin{equation}
r_\alpha\mapsto(r_\alpha^2-R_\alpha^2)^n
\end{equation}
itself and all its first $n-1$ derivatives with respect to $r_\alpha$ are equal to zero at $r_\alpha=R_\alpha$.
Consequently, we ensure smoothness on the boundary of the sphere,
\begin{equation}\label{smoothness_on_boundary}
\frac{\diff^k}{\diff r_\alpha^k} \rhoab_{L}(r_\alpha) \bigg\vert_{r_\alpha=R_\alpha} = 0 \quad\forall k=0,\ldots,n-1\, .
\end{equation}
Note that this includes localization of $\rhoab$ in $\sphere$ and thus the pseudo-charge density equals the charge density in I (condition 2 in Sect.~\ref{sec:WpschargeM}).

The smoothness of the pseudo-charge density is connected to the convergence properties of its Fourier series. 
Applying the Riemann-Lebesgue lemma for Fourier series to a $n-1$-fold differentiable function in combination with the differentiation rule for a Fourier transform, one can show that the Fourier coefficients for $K\to\infty$ go faster to zero than $\frac{1}{K^{n-1}}$, \textit{i.e.}\ we obtain the fastest convergence of the Fourier series for large $K$ and the convergence becomes the better the larger $n$ is.
If we choose $n$ too large, however, we are left with the small-$K$ Fourier coefficients only, and thus the Fourier series is unbalanced, in the sense that smaller coefficients have a larger weight.
So, ideally, our choice of $n$ is guided by the cut-off of the Fourier series.

For an explicit rule on how to choose $n$, Weinert~\cite{weinert1981solution} discussed the factors $Q_{Ln}^\alpha A_{\ell n}^\alpha(K) / \qlm[\rhoab]$, where in the Coulomb case 
\begin{equation}\label{Coulomb_factor}
Q_{Ln}^\alpha A_{\ell n}^\alpha (K) = \frac{j_{\ell +n+1}(KR_\alpha)}{K^{n+1}(2\ell+1)!!} \frac{(2\ell+2n+3)!!}{R_\alpha^{\ell +n+1}} \qlm[\rhoab]\, .
\end{equation}
His reasoning is based on the zeros of the $K$-dependent function $Q_{Ln}^\alpha A_{\ell n}^\alpha(K) / \qlm[\rhoab]$.  
Viewed as a function of $K$, however, our factor $Q_{Ln}^\alpha A_{\ell n}^\alpha(K) / \qlm[\rhoab]$ differs from Weinert's one only by a multiplicative constant. 
Thus Weinert's arguments apply here as well.
With Prop.~\ref{app:multiplicative_constant} in Appendix~\ref{app} it follows that this multiplicative constant is smaller than 1 for $\lambda>0$. 
We see this confirmed in Fig.~\ref{fig:weinertfactor}, which shows the $KR_\alpha$-dependence of the factor $Q_{Ln}^\alpha A_{\ell n}^\alpha(K)/\qlm[\rhoab]$ in the Yukawa~\eqref{Yukawa_factor}  and Coulomb~\eqref{Coulomb_factor} cases for several combinations of $n$ and angular degree $\ell$---it reveals a smaller amplitude in the Yukawa case. 
In agreement with Ref.~\onlinecite{weinert1981solution}, in Fig.~\ref{fig:weinertfactor} we make two observations:
(i) $Q_{Ln}^\alpha A_{\ell n}^\alpha / \qlm[\rhoab]$ as a function of $KR_\alpha$ has larger oscillations for smaller $n$ and 
(ii) for fixed $n$, the largest contribution to the Fourier series comes from $KR_\alpha$ less than the first zero of the Bessel function $j_{\ell +n+1}$. 
Since we deal with a finite number of $\qq$ vectors, we would like to reduce the oscillations mentioned in (i) by choosing a large $n$. 
On the other hand, due to the cut-off of the Fourier series at some $K_\text{max}$, the factor $Q_{Ln}^\alpha A_{\ell n}^\alpha / \qlm[\rhoab]$ must be small for $K>K_\text{max}$, which limits $\ell+n+1$ to a certain value, since the first zero is pushed towards infinity for growing $\ell+n+1$, as can be seen from a comparison between the pink and yellow, or the blue and purple lines in Fig.~\ref{fig:weinertfactor}. 
From this arises  Weinert's  criterion for choosing $n(\ell)$, which we adopt here:
\begin{itemize}
\item Choose $\nu\in\mathbb{N}$ such that the first zero of $j_\nu(z)$ is approximately equal to $(KR_\alpha)_\text{max}$.\label{criterion}
\item Then $n(\ell)$ is fixed by the relation $\nu=\ell +n+1$. 
\end{itemize}
Note that in this method $\ell$ is compensated by $n$ in such a way that $\nu$ is \textit{de facto} not depending on $\ell$. 
Since the discretization of $\qq$ vectors and the muffin-tin radii usually do not change over the course of the self-consistent-field iteration, the terms in~\eqref{rhoab_Fourier_final} depending on $\nu$ need to be computed just once.

\begin{figure}[htp]
\centering{
\includegraphics[width=0.5\textwidth]{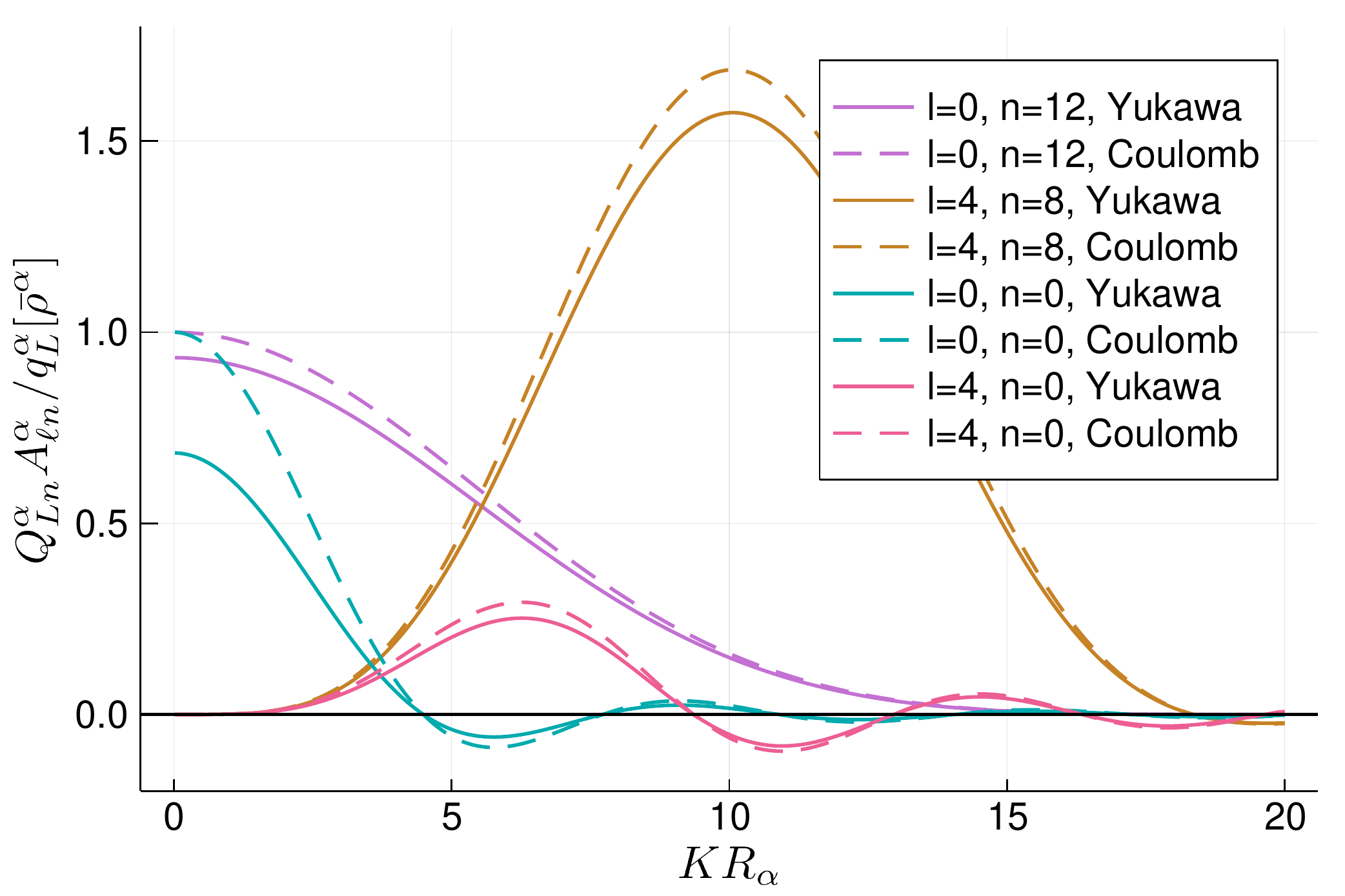}
\caption{ The factor $Q_{Ln}^\alpha A_{\ell n}^\alpha / \qlm[\rhoab]$ for several $(\ell,n)$ in the Yukawa and Coulomb cases with $R_\alpha=1$ and $\lambda=2$.  }\label{fig:weinertfactor}
}
\end{figure}

%% file: AlgorithmBulk.tex
\subsection{Algorithm: Construction of Yukawa Potential}\label{sec:Algorithm}

Algorithm~\protect\ref{alg:YukawaPotential} summarizes the construction of the Yukawa potential derived in this paper. 

\begin{figure} 
\begin{algorithm}[H] 
  \caption{Bulk-Case Yukawa Potential}
  \label{alg:YukawaPotential}
  \begin{algorithmic}[1]
    \Require charge density $\rho$, integer $\nu$ chosen as described on page~\protect\pageref{criterion} and preconditioning parameter $\lambda$.
    \Ensure Yukawa potential $\preV$ solving the modified Helmholtz equation~\protect\eqref{mod_Helmholtz_eq} with periodic boundary conditions. 
    \Statex \textbf{Pseudo-charge density} $\pseudo\leftarrow\rho$:     \State Modified multipole moments $\qlm[\rhoI]$ of the interstitial charge density in $\sphere$. \Comment{Eq.~\protect\eqref{Iqlm}}
    \State Modified multipole moments $\qlm[\rhoa]$ of the muffin-tin charge density in $\sphere$. \Comment{Eq.~\protect\eqref{MTqlm}}
    \State Modified multipole moments $\qlm[\rhoab] = \qlm[\rhoa] - \qlm[\rhoI]$ of $\rhoab$.
    \State Sphere-localized part $\rhoab(\qq)$ of the pseudo-charge density. \Statex\Comment{Eqs.~\protect\eqref{rhoab_Fourier_final} and~\protect\eqref{rhoab0}}
    \State Pseudo-charge density $\rhot(\qq)$. \Comment{Eq.~\protect\eqref{rhot_Fourier}}
    \Statex \textbf{Interstitial potential} $\preV^\text{I}$ $\leftarrow\pseudo$:
    \State Interstitial potential $\preV^\text{I}(\qq)$. \Comment{Eq.~\protect\eqref{preVI_Fourier}}
    \Statex \textbf{Muffin-tin potential} $\preV^\text{MT}\leftarrow\rho^\text{MT},\preV^\text{I}$:
    \State Boundary terms $\VIR$ of muffin-tin potential. \Statex\Comment{Eq.~\protect\eqref{spherical_harmonics_boundary_term}}
    \State Radial parts $V_{L}^\alpha(r_\alpha)$ of muffin-tin potential. \Comment{Eq.~\protect\eqref{computable_MTPot}}
  \end{algorithmic}
\end{algorithm}
\end{figure}

In the case that Weinert's method is available as an implemented algorithm then only relatively few changes are necessary to make it available for the solution of the modified Helmholtz equation.
The changes to be made in practice are limited to the following: 
The slightly different radial behavior of the Green function leads to small changes in the multipole moments of the interstitial and muffin-tin charge densities in each sphere, $\qlm[\rhoI]$~\eqref{Iqlm} and $\qlm[\rhoa]$~\eqref{MTqlm} respectively, and in the  Fourier components of the pseudo-charge density $\bar{\rho}^\alpha(\qq)$~\eqref{rhoab_Fourier_final} and~\eqref{rhoab0}. 
The integer $\nu$ in the formula for the pseudo-charge density's Fourier components, which determines the convergence of the Fourier series, is chosen exactly the same as in Weinert's original method.
The interstitial potential~\eqref{preVI_Fourier} undergoes changes indirectly through the pseudo-charge density and directly by the prefactor $\frac{4\pi}{K^2+\lambda^2}$ that substitutes $\frac{4\pi}{K^2}$. 
Since the $\qq=\nullvec$-term is well-defined, it is not set to a constant as in the original method.
The muffin-tin potential is affected only in its radial dependence in both the boundary and the source contribution of~\eqref{computable_MTPot}. Basically, the polynomials $r^\ell$ and $1/r^{\ell+1}$ in these quantities are substituted by the modified spherical Bessel functions $i_\ell(\lambda r)$ and $k_\ell(\lambda r)$, respectively, and the prefactor $\frac{4\pi}{2\ell+1}$ is substituted by $4\pi\lambda$.
The interstitial potential on the boundary of the spheres, $\VIR$~\eqref{spherical_harmonics_boundary_term}, only changes indirectly through the changed values of $\preV^\text{I}(\qq)$.

%% file: Conclusion.tex
\section{Conclusion and Outlook}\label{sec:conclusion}

We have presented a general method for solving the modified Helmholtz equation for a 3D-periodic system of charge densities not restricted by any shape approximation of three-dimensional volume. The three-dimensional domain is typically decomposed into non-overlapping atom-centered spheres (the muffin-tin region) and the space between these spheres (the interstitial region). The solution, the Yukawa potential, suffices 3D-periodic boundary conditions. Since the Yukawa differential equation is similar to the Poisson equation, we leveraged our derivations on the work of Weinert~\cite{weinert1981solution}. Our work can be considered an extension of Weinert's work determining instead of the bare Coulomb potential with zero screening, the screened Coulomb potential with finite screening length $1/\lambda$. Like Weinert's pseudo-charge  method, our extension is based on the concept of the non-uniqueness of the multipole expansion as well as  on the Dirichlet boundary value problem applied to a sphere.  
The difference between the modified Helmholtz equation and the Poisson equation lies solely in the radial behavior and thus the homogeneous solutions to the radial part of the differential equation are modified spherical Bessel functions instead of polynomial functions. The consequence is a different radial behavior of the Green function, resulting in the screening of the potential, which is now expanded in modified multipole moments, and this in turn affects the pseudo-charge density. Furthermore, the modification of the multipole moments implies that the modified monopole is not connected anymore to the total charge. We have shown that Weinert's convergence analysis of an absolutely and uniformly convergent Fourier series of the pseudo-charge density is transferred to the modified pseudo-charge density and thus  we  can therefore best choose the same integer parameters for convergence. 
Finally we layed out the minor changes necessary to change an  implemented method for solving the Poisson equation available to an implementation for solving  the modified Helmholtz equation.

Considering that Weinert's pseudo-charge method has become the standard method for calculating the electrostatic potential without shape approximation in all-electron band structure methods for applications of periodic solids, and since we have extended it to a modified pseudo-charge method with only minor modifications involving some radial integrals, this now allows to treat the screened Coulomb potentials without shape approximation described by the modified Helmholtz equation with all-electron methods. The screened Coulomb or Yukawa potential typically occurs in single particle or mean-field theories to the problem of many charged particles, where all charged particles contribute effectively to the screening of the bare Coulomb potential.

%% file: Acknowledgments.tex
\section{Acknowledgments}

We would like to thank Gregor Michalicek, Jan Winkelmann and Rudolf Zeller for fruitful discussions and help with conceptual and computational matters. This work has been supported by a JARA-HPC seed-fund project and by the MaX Center of Excellence funded by the EU through the H2020-EINFRA-2015-1 project: GA 676598.
The authors gratefully acknowledge the computing time granted by the JARA-HPC Vergabegremium and provided on the JARA-HPC partition’s part of the supercomputer JURECA~\cite{jureca} at Forschungszentrum J\"ulich.

%% file: Appendix.tex
\appendix
\section{}\label{app}

\begin{prop}\label{app:1}
Let $R>0$, $\kappa>0$ and $f_\ell$ be either the spherical Bessel function of the first kind, $f_\ell=j_\ell$, or the modified spherical Bessel function of the first kind, $f_\ell=i_\ell$. Then 
\begin{equation*}\label{bessel_integration_property}
\int_0^R (r^2-R^2)^n r^{\ell+2} f_\ell(\kappa r) \diff r = (-2)^n n! R^{\ell+n+2} \frac{f_{\ell+n+1}(\kappa R)}{\kappa^{n+1}}
\end{equation*}
holds for all $n\in\mathbb{N}_0$ and all $\ell\in\mathbb{N}_0$.
\end{prop}

\begin{proof}[Proof by mathematical induction on $n$]
The base case $n=0$ follows immediately from the differentiation property~\eqref{bessel_identity} of the functions $j_\ell$ and $i_\ell$, by
\begin{align*}
    \int_0^{R} r^{\ell+2} f_\ell(\kappa r) \diff r
    &=\kappa^{-\ell-3}\int_0^{\kappa R} r^{\ell+2} f_\ell(r) \diff r \\
    &=\kappa^{-1} R^{\ell+2} f_{\ell+1}(\kappa R)\,.
\end{align*}
The induction step uses the property for $n-1$ and $\ell+1$, and partial integration with $u(r)=(r^2-R^2)^n$ and $v^\prime(r)=r^{\ell+2}f_\ell(\kappa r)$, to derive the statement for $n$ and $\ell$. 
Let 
\[F(n,\ell)= \int_0^R (r^2-R^2)^n r^{\ell+2} f_\ell(\kappa r) \diff r\]
be the left-hand side of~\eqref{bessel_integration_property}.
Then
\begin{align*}
    F(n,\ell)
    &=[(r^2-R^2)^n \kappa^{-1} r^{\ell+2} f_{\ell+1}(\kappa r)]_0^{R} \nonumber \\
    &- \int_0^{R} n (r^2-R^2)^{n-1} 2r \kappa^{-1} r^{\ell+2} f_{\ell+1}(\kappa r) \diff r \\
    &= -2n\kappa^{-1} F(n-1,\ell+1) \\
    &= (-2)^n n! R^{\ell+n+2} \frac{f_{\ell+n+1}(\kappa R)}{\kappa^{n+1}}\, ,
\end{align*}
where we used the induction hypothesis in the last equation, and thus the proposition follows.
\end{proof}

\begin{prop}\label{app:multiplicative_constant}
For nonnegative $\lambda$ and $R$, and $\nu\in\mathbb{N}$
\begin{equation*}
\frac{i_\nu(\lambda R)}{\lambda^\nu}\geq \frac{R^\nu}{(2\nu+1)!!}
\end{equation*}
holds, with equality if and only if $\lambda=0$ or $R=0$.
\end{prop}

\begin{proof}
Since $i_\nu(x)=\iu^{-\nu}j_\nu(\iu x)$ and $j_\nu$ has the expansion~\cite{arfken2013mathematical}
\begin{equation*}
j_\nu(x) = \frac{x^\nu}{(2\nu+1)!!}\sum_{s=0}^\infty \frac{(-1)^s}{s!(\nu+\frac{3}{2})_s} \left(\frac{x}{2}\right)^{2s}\, ,
\end{equation*}
where $(\cdot)_s$ is the Pochhammer symbol defined for general $a\in\mathbb{R}$ and $s\in\mathbb{N}_0$ by
\begin{equation*}
(a)_0 = 1\,,\quad (a)_s = a (a+1) \cdots (a+s-1)\,,
\end{equation*}
$i_\nu$ can be expanded by
\begin{equation*}
i_\nu(x) = \frac{x^\nu}{(2\nu+1)!!}\sum_{s=0}^\infty \frac{1}{s!(\nu+\frac{3}{2})_s} \left(\frac{x}{2}\right)^{2s}\, .
\end{equation*}
$\frac{i_\nu(\lambda R)}{\lambda^\nu}$ thus becomes
\begin{equation*}
\frac{i_\nu(\lambda R)}{\lambda^\nu} = \frac{R^\nu}{(2\nu+1)!!}\sum_{s=0}^\infty \frac{1}{s!(\nu+\frac{3}{2})_s} \left(\frac{\lambda R}{2}\right)^{2s}
\end{equation*}
and it remains to show that 
\begin{equation*}
    \sum_{s=0}^\infty \frac{1}{s!(\nu+\frac{3}{2})_s} \left(\frac{\lambda R}{2}\right)^{2s}\geq 1\,.
\end{equation*}
The first term of the sum for $s=0$ is always $1$ irrespective of the values of $\lambda$ and $R$. 
The other summed terms for $s>0$ are positive, if both $\lambda$ and $R$ are positive, and zero, if one of the two variables is zero.
\end{proof}

%% file: DataAvailability.tex
\smallskip
\section*{Data Availibility}

Data sharing is not applicable to this article as no new data were created or analyzed in this study.